%% file: 2015-ARXIV-AdaptiveKlee-BarbayPerezRojas.tex
\documentclass[orivec]{llncs}
\usepackage{llncsdoc}
\usepackage{amsmath,amssymb}
\usepackage{graphicx}
\usepackage{epsfig}
\usepackage{caption}
\usepackage{setspace}
\usepackage{subcaption}
\usepackage{hyperref}
\usepackage{algorithm}
\usepackage[noend]{algpseudocode}
\usepackage{enumerate}
\usepackage{versions}
\usepackage{mathtools} 
	
\excludeversion{INUTILE}
\includeversion{LONG}
\excludeversion{SHORT}
\includeversion{CONST-TW}
\markversion{TODO}

\DeclareMathOperator{\polylog}{polylog}

\newcommand{\bigo}[1]{\mathcal{O}#1}

\begin{document}

\input{AdaptiveKlee}

\bibliographystyle{splncs}
\bibliography{references}

\end{document}

%% file: AdaptiveKlee.tex
\hyphenation{pro-blems}
\pagestyle{headings}  
\addtocmark{Adaptive Computation of the Klee's Measure in High Dimensions}

\mainmatter              
\title{Adaptive Computation of   the Klee's Measure\\ in High Dimensions}
\titlerunning{Adaptive Computation of the Klee's Measure\\ in High Dimensions}
\author{
  J\'er\'emy Barbay\inst{1} 
  \and
  Pablo P\'erez-Lantero\inst{2}
  \and 
  Javiel Rojas-Ledesma\inst{1}\thanks{Corresponding Author}
}

\institute{
  Departmento de Ciencias de la Computaci\'on, Universidad de Chile, Chile\\ 
  \email{jeremy@barbay.cl, jrojas@dcc.uchile.cl.}
  \and
  Escuela de Ingenier\'ia Civil Inform\'atica, Universidad de Valpara\'iso, Chile.\\ 
  \email{pablo.perez@uv.cl.}
}


\maketitle

\begin{abstract} 
The \textsc{Klee's Measure} of $n$ axis-parallel boxes in $\mathbb{R}^d$ is the volume of their union. It can be computed in time within $\bigo(n^{d/2})$ in the worst case. We describe three techniques to boost its computation: one based on some type of ``degeneracy'' of the input, and two ones on the inherent ``easiness'' of the structure of the input. 
The first technique benefits from instances where the \textsc{Maxima} of the input is of small size $h$, and yields a solution running in time within $\bigo(n\log^{2d-2}{h}+ h^{d/2}) \subseteq \bigo(n^{d/2}$).  
The second technique takes advantage of instances where no $d$-dimensional axis-aligned hyperplane intersects more than $k$ boxes in some dimension, and yields a solution running in time within $\bigo(n \log n + n k^{(d-2)/2}) \subseteq \bigo(n^{d/2})$.
The third technique takes advantage of instances where the \emph{intersection graph} of the input has small treewidth $\omega$. It yields an algorithm running in time within $\bigo(n^4\omega \log \omega + n (\omega \log \omega)^{d/2})$ in general, and in time within $\bigo(n \log n +  n \omega ^{d/2})$ if an optimal tree decomposition of the intersection graph is given.
We show how to combine these techniques in an algorithm which takes advantage of all three configurations.
\end{abstract}

   \section{Introduction}\label{sec:Introduction}				  

The \textsc{Klee's Measure} of a set of $n$ axis-parallel boxes in $\mathbb{R}^d$ is defined as the volume of the union of the boxes in the set~\cite{Bentley1977}. Its computation was first posed by Victor Klee in 1977~\cite{Klee1977}, who originally considered the measure for intervals in the real line. Bentley~\cite{Bentley1977} generalized the problem to $d$ dimensions and described an algorithm running in time within $\bigo(n^{d-1}\log n)$. Several years later, Overarms and Yap~\cite{Overmars1991} described a solution running in time within $\bigo(n^{d/2}\log n)$, which remained essentially unbeaten for more than 20 years until 2013, when Chan~\cite{Chan2013} presented an algorithm running in time within $\bigo(n^{d/2})$. We consider that, additionally, a $d$-dimensional domain box $\Gamma$ is given, making the objective to compute the \textsc{Klee's Measure} within $\Gamma$.

Some special cases of this problem have been studied, such as the \textsc{Hypervolume} problem, where boxes are orthants of the form $\{(x_1,\ldots,x_d) \in \mathbb{R}^d \mid (x_1~\le~\alpha_1)  \land \ldots \land (x_d~\le~\alpha_d) \}$, each $\alpha_i$ being a real number, which can be solved in time within $\bigo(n^{d/3}\polylog n)$;
and \textsc{Cube-KMP}~\cite{Agarwal2010,Bringmann2012}, when the boxes are hypercubes, which can be solved in running time within $\bigo(n^{(d+1)/3}\polylog n)$~\cite{Chan2013}.
\begin{LONG}
Yildiz and Suri\cite{Yildiz2012} considered \textsc{$k$-Grounded-KMP}, the case when the projection of the input boxes to the first $k$ dimensions is an instance of  \textsc{Hypervolume}. They described an algorithm to solve \textsc{$2$-Grounded} in time within $\bigo(n^{(d-1)/2} \log^2 n)$, for any dimension $d \ge 3$.
\end{LONG}

The best lower bound known for the computational complexity of the \textsc{Klee\rq{}s Measure} problem in the worst case over instances of size $n$ is within $\Omega(n \log n)$, so far tight only for dimensions one and two~\cite{Fredman78} as the best known upper bound is $\bigo(n^{d/2})$ in dimension $d$. Chan~\cite{Chan08} conjectured that no \lq{}purely combinatorial\rq{} algorithm computing the  \textsc{Klee\rq{}s Measure} in dimension $d$ exists running in time within $\bigo(n^{d/2-\varepsilon} )$ for any $\varepsilon > 0$.  
\begin{LONG}
He proved that if the $d$-dimensional \textsc{Klee\rq{}s Measure} problem can be solved in time $T_d(n)$, then one can decide whether an arbitrary $n$-vertex graph $G = (V , E)$ contains a clique of size $d$ in  time within $\bigo(T_d(\bigo(n^2)))$. The current best combinatorial algorithm for finding $k$-cliques in a graph, requires near-$\bigo(n^k)$ time, and hence the conjecture.
\end{LONG}

In an adaptive analysis, the cost of an algorithm is measured as a function of, not just the input size, but of other parameters that capture the inherent simplicity or difficulty of an input instance~\cite{Afshani2009}. An algorithm is said to be adaptive if ``easy'' instances are solved faster than the ``hard'' ones. There are adaptive algorithms to solve classical problems such as \textsc{Sorting} a permutation~\cite{MoffatP92}, \textsc{Sorting} a multiset~\cite{Barbay2013}, computing the \textsc{Convex Hull}~\cite{Kirkpatrick1986} of a set of points in the plane and in 3-space, and computing the \textsc{Maxima} of a set of $d$-dimensional vectors~\cite{Kirkpatrick1985}. There are also adaptive algorithms for the \textsc{Maximum Weight Box} problem~\cite{BarbayCNP14}, of particular interest since, for any dimension $d \ge 2$, the \textsc{Maximum Weight Box} problem can be reduced to an instance of the \textsc{Klee\rq{}s Measure} problem in $2d$ dimensions. 

Even though the asymptotic complexity of $\bigo(n^{d/2})$ is the best known so far for the \textsc{Klee's Measure} problem~\cite{Chan2013}, there are many cases which can be solved in time within $\bigo(n\lg n)$ (see Figures~\ref{fig:degeneracy}~and~\ref{fig:degeneracy_comp} for some examples).  Some of those ``easy'' instances can be mere particular cases, some others can be hints of some hidden measures of difficulty of the \textsc{Klee's Measure} problem. 
 
 \paragraph{\textbf{Hypothesis}:}There are such difficulty measures that gradually separate instances of the same size $n$ into various classes of difficulty; from easy ones solvable in time within $\bigo(n \log n)$, to difficult ones which the best known algorithm solves in time within $\bigo(n^{d/2})$.

\begin{figure}
\centering
\begin{subfigure}{.45\textwidth}
  \centering
  \includegraphics[height=4cm]{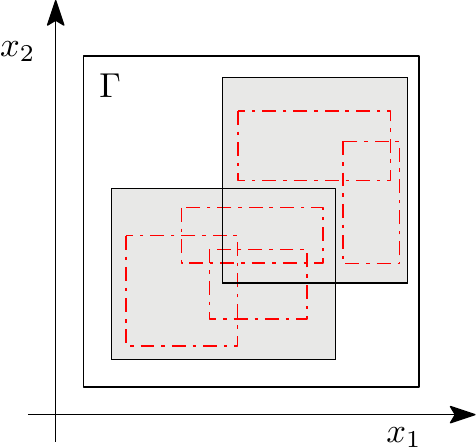}
  \caption{ }
  \label{fig:degeneracy_boxes}
\end{subfigure}
\begin{subfigure}{.45\textwidth}
  \centering
  \includegraphics[height=4cm]{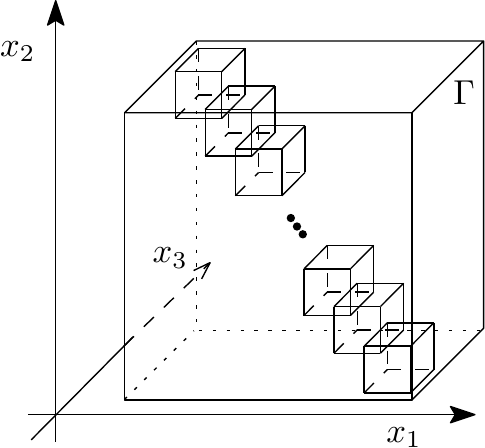}
  \caption{ }
  \label{fig:degeneracy_orthants}
\end{subfigure}
\caption{Two \lq{}easy\rq{} instances of the \textsc{Klee's Measure} problem: red dashed boxes in (a) can be removed without affecting the \textsc{Klee's Measure} within the domain $\Gamma$, yielding a much smaller instance to solve; while the instance in (b) belongs to the family of instances which intersection graph is a tree (a path in this particular case), that can be solved in time within $\bigo(n \log n)$ by a divide-and-conquer algorithm.}
\label{fig:degeneracy}
\end{figure}

\paragraph{\textbf{Results}:}

We describe three techniques to boost the computation of the \textsc{Klee's Measure} on ``easy'' instances, and analyze each in the adaptive model. For each technique, we identify a proper difficulty measure, which models the features which the technique is taking advantage of.  The first technique is the simplest, taking advantage of degenerated instances, while the second and third ones are more sophisticated.

The first technique (described in Section~\ref{sec:maximaFiltering}) is related to a classical problem in Computational Geometry: the computation of the \textsc{Maxima} of a set of vectors. A vector in a set $\mathcal{T} \subset \mathbb{R}^d$ is called \emph{maximal} if none of the remaining vectors in $\mathcal{T}$ dominates it in every component. The \textsc{Maxima} of $\mathcal{T}$ (denoted by $M(\mathcal{T})$) is the set of maximal elements in $\mathcal{T}$. In 1985, Kirkpatrick and Seidel~\cite{Kirkpatrick1985} gave an output-size sensitive algorithm for this problem, running in time within $\bigo(n \log^{d-2}{h})$, where $h$ is the size of the \textsc{Maxima}.  We extend the concept of \textsc{Maxima} to the sets of boxes, and describe an algorithm computing the \textsc{Klee\rq{}s Measure} in time within $\bigo(n\log^{2d-2}{h}+ h^{d/2}) \subseteq \bigo(n^{d/2})$, where $h$ denotes the size of the \textsc{Maxima} of the input set.

The second technique (described in Section~\ref{sec:intrinsic}) is based on the \emph{profile} $k$ of the input set, which D'Amore et al. \cite{dAmoreNRW95} defined as the minimum, over all dimensions $i$, of the maximum number of boxes intersected by a same axis aligned hyperplane orthogonal to $i$.  The algorithm described by Chan~\cite{Chan2013} to compute the \textsc{Klee's Measure} is in fact sensitive to a difficulty measure based on a slightly weaker definition of the profile. We improve on this by describing an algorithm to compute the \textsc{Klee's Measure} in time within $\bigo(n \log n + nk^{(d-2)/2})\subseteq \bigo(n^{d/2})$, where $k$ is the profile of the input set. 

The third technique (described in Section~\ref{sec:intersectionGraphTreewidth}) is based on the \emph{treewidth} of the \emph{intersection graph} of the input set. The intersection graph of a set of boxes is a graph $G$ where the vertices are the boxes, and where two boxes are connected by an edge if and only if they intersect. The treewidth $\omega$ of a graph measures how ``close'' the graph is to a tree. This technique yields a solution running in time within $\bigo(n \log n +  n \omega ^{d/2})$ if a tree decomposition of the intersection graph of the input set is given; and a solution running in time within $\bigo(n^4\omega \log \omega + n (\omega \log \omega)^{d/2})$ if only the boxes are given.

In Section~\ref{sec:combining}, we discuss how to compare and combine these three techniques, and in Section~\ref{sec:discussion} we describe some potential directions for future work. 

\section{Maxima Filtering}\label{sec:maximaFiltering}

Our first technique considers the \textsc{Maxima} of the input set of boxes to take advantage of instances where many boxes can be ``filtered out'' in small time.  
\begin{SHORT}The technique is simple:  we describe it shortly for completeness (see the extended version~\cite{BarbayPR15} for more details).\end{SHORT}
A box in a set $\mathcal{B}$ is called \emph{maximal} if none of the remaining boxes in $\mathcal{B}$ completely contains it. The \textsc{Maxima} $M(\mathcal{B})$ of $\mathcal{B}$ is the set of maximal elements in $\mathcal{B}$. One can observe that, by definition, elements not in the \textsc{Maxima} of an input set of the \textsc{Klee's Measure} problem can be removed from the input set without affecting the value of the \textsc{Klee's Measure}. 

\begin{LONG}
Algorithm~\ref{alg:maximaadpt} takes advantage of this fact to compute the \textsc{Klee's Measure} in time sensitive to the size of the \textsc{Maxima} of the input set.

\begin{algorithm}                      
\caption{\texttt{maxima\_adaptive\_measure}}          
\label{alg:maximaadpt}                           
\begin{algorithmic}[1]
\Require A $d$-dimensional domain box $\Gamma$, and a set of $n$ $d$-dimensional boxes $\mathcal{B}$
\Ensure The \textsc{Klee\rq{}s Measure} of $\mathcal{B}$ within $\Gamma$
\State Compute $M(\mathcal{B})$, the \textsc{Maxima} of $\mathcal{B}$
\State  \Return \texttt{SDC}($\Gamma, M(\mathcal{B})$)
\end{algorithmic}
\end{algorithm}
$\;$\end{LONG}
Overmars~\cite{Overmars81} showed that if the \textsc{Maxima} of a set of $n$ $d$-dimensional vectors can be computed in time $T_d(n)$, then the \textsc{Maxima} of a set of $n$ boxes can be computed in time within $\bigo(T_{2d}(n))$. 
\begin{LONG}
To prove this, Overmars~\cite{Overmars81} expressed each box $b_i = [l_{i,1}, u_{i,1}] \times \ldots \times [l_{i,d}, u_{i,d}]$ as a $2d$ dimensional vector $\vec{b_i}$~$=(-l_{i,1}, u_{i,1}, \ldots, -l_{i,1}, u_{i,d})$. Note that if $b_i, b_j$ are boxes, then $b_i$ dominates $b_j$ if and only if $\vec{b_i}$ dominates $\vec{b_j}$. 
\end{LONG} 
We use this result to show in Lemma~\ref{lemma:kmp_boxes} that the \textsc{Klee's Measure} can be computed in running time sensitive to the size of \textsc{Maxima} of the input.

\begin{lemma}\label{lemma:kmp_boxes}
Let $\mathcal{B}$ be a set of $n$ boxes in $\mathbb{R}^d$ and $\Gamma$ a d-dimensional box. The \textsc{Klee's Measure} of $\mathcal{B}$ within $\Gamma$ can be computed in time within $\bigo(n(\log h)^{2d-2} + h^{d/2})$, where $h$ is the size of the \textsc{Maxima} $M(\mathcal{B})$ of $\mathcal{B}$.
\end{lemma}
\begin{proof}
\begin{SHORT}
The bound can be achieved by composing two known algorithms: ($i$) As a consequence of Overmars' result~\cite{Overmars81}, the \textsc{Maxima} $M(\mathcal{B})$ of $\mathcal{B}$ can be computed in time within $\bigo(n \log^{2d-2} {h})$ using the output-size sensitive algorithm described by Kirkpatrick and Seidel~\cite{Kirkpatrick1985}; ($ii$) the \textsc{Klee's Measure} of $M(\mathcal{B})$, of size $h$, can be computed in time within $\bigo(h^{d/2})$ using the algorithm  proposed by Chan~\cite{Chan2013}.  The time bound follows.
\end{SHORT}
\begin{LONG}
Algorithm~\ref{alg:maximaadpt} achieves the bound given in the lemma: as a consequence of Overmars' result~\cite{Overmars81}, the \textsc{Maxima} $M(\mathcal{B})$ of $\mathcal{B}$ is computed in step one in time within $\bigo(n \log^{2d-2} {h})$ using the output-size sensitive algorithm described by Kirkpatrick and Seidel~\cite{Kirkpatrick1985}; in the second step, the \textsc{Klee's Measure} of $M(\mathcal{B})$, of size $h$, is computed in time within $\bigo(h^{d/2})$ using the algorithm  proposed by Chan~\cite{Chan2013}. The result follows.
\end{LONG}

\qed
\end{proof}
\begin{LONG}
Note that in the bound from the Lemma~\ref{lemma:kmp_boxes},  the base with exponent $d/2$ is $h$, instead of $n$ as in the bound $\bigo(n^{d/2})$ for the running time of $\texttt{SDC}$.  In degenerated instances, where $h$ is significantly smaller than $n$, the bound from Lemma~\ref{lemma:kmp_boxes} is significantly better than $\bigo(n^{d/2})$. 
\end{LONG}

One way to further improve this result is to remove dominated elements at each recursive call of Chan's algorithm~\cite{Chan2013}: we discuss the difficulties in analyzing this approach in Section~\ref{sec:discussion}. In the next section we describe another boosting technique, which still reduces to Chan $\bigo(n^{d/2})$'s algorithm, but is less focused on degenerated instances.

\section{Profile-based Partitioning}\label{sec:intrinsic}

The \emph{$i$-th profile} $k_i$ of a set of boxes $\mathcal{B}$ is defined as the maximum number of boxes intersected by any hyperplane orthogonal to the $i$-th dimension. The \emph{profile} $k$ of a set of boxes is defined as $k=\min_{i\in[1..d]}\{k_i\}$. D'Amore et al.~\cite{dAmoreNRW95} showed how to compute it in linear time (after sorting the boxes in each dimension).
We make the observation that Chan's algorithm~\cite{Chan2013} for this problem is adaptive to a measure slightly different to the profile (and weaker than it); and improve on this result by describing a technique which yields a solution sensitive to the profile of $\mathcal{B}$. 

\subsection{Intrinsic Adaptivity of Chan\rq{}s Algorithm}

The \textit{Simplify, Divide and Conquer} algorithm (\texttt{SDC} for short)  proposed by Chan~\cite{Chan2013} 
to compute the \textsc{Klee\rq{}s Measure}, already behaves adaptively in the sense that it runs faster on some large families of instances. Let the \emph{quasi-profile} $\kappa$ of a set of boxes be defined as $\kappa= \max \{ k_i \mid i \in [1..d] \}$, where $k_i$ denotes the $i$-th profile. 
\begin{LONG}
\paragraph{\textbf{Observation}}

Let $B$ be a set of boxes having quasi-profile $\kappa$ within a domain box $\Gamma$.  Algorithm \texttt{SDC} computes the \textsc{Klee's Measure} of $\mathcal{B}$ within $\Gamma$ in time within  $\bigo(n \log n + n\kappa^{(d-2)/2})$

The proof of this observation is quite technical and long. Since the result in next section subsumes this one, and the analysis there is considerably simpler, we omit the proof of this observation. 
\end{LONG}
\begin{SHORT}
(see the technical  proof in Section 3.1 of the extended version~\cite{BarbayPR15}).
\end{SHORT}

An example of the class of instances with small quasi-profile is illustrated in Figure~\ref{fig:degeneracy_orthants}. In the next subsection, we describe a slightly modified version of the algorithm  \texttt{SDC} which runs in time sensitive to the profile $k$ of $\mathcal{B}$ rather than its quasi-profile $\kappa$, an improvement since $k \le \kappa$ on all instances.

\subsection{Profile-based partitioning}\label{subsec:profile-based-partition}

Let $\mathcal{B}$ be a set of boxes with profile $k$, and $\Gamma$ a domain box. Given the profile $k$ of $\mathcal{B}$, Algorithm~\ref{alg:splitdomain} splits $\Gamma$ into $m \in \bigo(n/k)$ slabs $\Gamma_1 \ldots \Gamma_m$, such that the measure of $\mathcal{B}$ within $\Gamma$ is equal to the summation of the measures of $B$ within $\Gamma_1,\ldots,\Gamma_m$, respectively. The algorithm performs a plane sweep by one of the dimensions with the smallest profile and cuts the domain by a hyperplane every $2k$ endpoints. By computing the \textsc{Klee's Measure} of $\mathcal{B}$ within each $\Gamma_i$, and summing up all those values, one can compute the \textsc{Klee's Measure} of $\mathcal{B}$ within $\Gamma$. 

\begin{algorithm}
\caption{\texttt{split-domain}}
\label{alg:splitdomain}
\begin{algorithmic}[1]
\Require A domain $\Gamma$, a set of $n$ boxes $\mathcal{B}$, and the profile $k$ of $\mathcal{B}$
\Ensure A partition of $\Gamma$ into $m$ slabs, intersecting each one $\bigo(k)$ boxes.
\State let $i$ be a dimension where the $i$-th profile $k_i$ of $\mathcal{B}$ equals $k$
\For {$j=1, 2, \ldots, m \in \bigo(n/k)$}
\State let $p \gets (2k \times j)$-th endpoint of $\mathcal{B}$ within $\Gamma$
\State split $\Gamma$ into $\{\Gamma_L, \Gamma_R\}$ by the hyperplane $x_i = p$
\State let $\Gamma_j \gets \Gamma_L$, and $\Gamma \gets \Gamma_R$
\EndFor
\State \Return $\{\Gamma_1, \ldots, \Gamma_m\}$
\end{algorithmic}
\end{algorithm}

Each of the slabs into which $\Gamma$ is divided in Algorithm~\ref{alg:splitdomain} can intersect at most $\bigo(k)$ boxes of $\mathcal{B}$: by definition of the profile, at most $\bigo(k)$ boxes can intersect the boundaries of the slab, and since each slab contains at most $2k$ endpoints, no more than $\bigo(k)$ boxes can completely lie in its interior. This can be used to bound the running time of the computation of the \textsc{Klee's Measure} of $\mathcal{B}$.

\begin{lemma}\label{lemma:kbound}
Let $\mathcal{B}$ be a set of $n$ boxes in $\mathbb{R}^d$, $\Gamma$ be a d-dimensional domain box, and $k$ denote the profile of $\mathcal{B}$. The \textsc{Klee's Measure} of $\mathcal{B}$ within $\Gamma$ can be computed in time within $\bigo \left(n \log n + nk^{\frac{d-2}{2}} \right)$.
\end{lemma}
\begin{proof}
Using Algorithm~\ref{alg:splitdomain}, one can split the domain into $\bigo(n/k)$ slabs in linear time after sorting the input. The measure within each slab can be computed in time within $\bigo(k^{d/2})$ using the algorithm \texttt{SDC}. The result follows.
\qed
\end{proof}

\begin{LONG}
Note that again in the bound from Lemma~\ref{lemma:kbound}, the value with $d$ in the exponent is the profile $k^*$, instead of the size of the set $n$.  Over instances with small profile $k^*$ (like the ones in the class illustrated in Figure~\ref{fig:degeneracy}b), the bound from Lemma~\ref{lemma:kbound} is significantly better than the upper bound $\bigo(n^{d/2})$ for the running time of $\texttt{SDC}$.
\end{LONG}
In the next section, we describe a technique, based on the treewidth of the intersection graph of the input set, a measure that captures how ``close'' a graph is to a tree. This technique takes advantage of inputs where the intersection graph is of small treewidth. 

\section{Intersection Graph's Treewidth}\label{sec:intersectionGraphTreewidth}
In instances such as the one described in Figure~\ref{fig:degeneracy_orthants}, where the intersection graph is a tree, a minor variant of \texttt{SDC} performs in time within $\bigo(n \log n)$ independently of the dimension $d$. 
\begin{LONG}
The concept of treewidth was discovered independently several times under different names (for a nice introduction, see Sections 10.4 and 10.5 of Kleinberg and Tardos' book~\cite{KleinbergTardos2005}). Many graph problems that are \texttt{NP}-hard for general graphs can be solved in polynomial time for graphs with small treewidth. For example, Arnborg and Proskurowski~\cite{Arnborg198911} showed that for most NP-hard problems that have linear time algorithms for trees, there are algorithms solving them in time  linear in the size of the graph but exponential or super-exponential in the treewidth. They illustrated the idea with classical optimization problems involving independent sets, dominating sets, graph coloring, Hamiltonian circuits and network reliability. 
\end{LONG}

In this section we describe how to generalize this behavior to instances with a more general intersection graph, taking advantage of the treewidth~\cite{KleinbergTardos2005} of this intersection graph. We recall the definition and some basic results on treewidth in Section~\ref{subsec:treewidth-preliminaries}, to apply them to the computation of the \textsc{Klee's Measure}in Section~\ref{subsec:treewidth-adapt}.

\subsection{Preliminaries}\label{subsec:treewidth-preliminaries}

The most widely used treewidth definition, based on \emph{tree decompositions}, was introduced by Robertson and Seymour~\cite{Robertson1986309}. 

\begin{definition} A tree decomposition of a graph $G=(V,E)$ is a pair $(\{X_i | i \in I\}, T=(I,F))$, with $\{X_i | i \in I\}$ a family of subsets of $V$ and $T$ a tree, such that:
\begin{itemize}
\item (Node coverage) $\bigcup_{i \in I} X_i = V$, 
\item (Edge coverage) for all $\langle u,v \rangle \in E$, there is an $i \in I$ with $u,v \in X_i$, and
\item (Coherence) if $v \in X_i \cap X_j$, then for all $k$ in the simple path from $i$ to $j$ in $T$ we have $v \in X_k$.
\end{itemize}
\end{definition}

One refers to the elements of $I$ as \emph{nodes}, and to the elements of $V$ as \emph{vertices}. The \emph{width} of a tree decomposition $(\{X_i | i \in I\}, T=(I,F))$ is $\max_{i \in I} |X_i|-1$. A tree decomposition of $G$ is called \emph{optimal} if  its width is the minimum width among all tree decompositions of $G$. The \emph{treewidth} $\omega$  of a graph $G$ is the width of an optimal decomposition of itself.
\begin{LONG}
  (see Figure~ \ref{fig:treewidth} for an illustration of tree decompositions and treewidth). 

\begin{figure}
\centering
  \includegraphics[height=3.2cm]{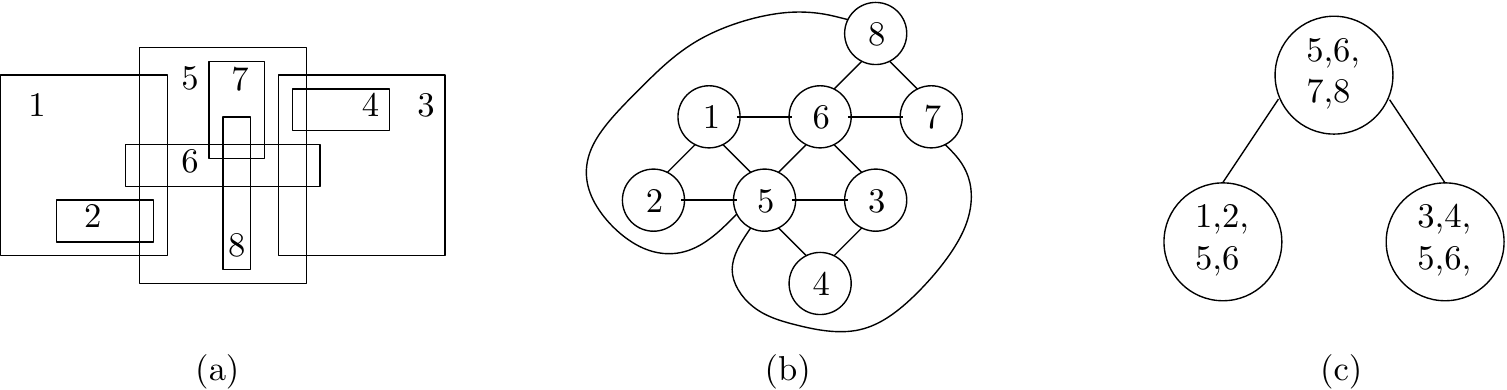}
  \caption{Tree decomposition: (a) a set of boxes; (b) the intersection graph of the set; and (c) an optimal tree decomposition of the graph, with treewidth $\omega = 2$ }
  \label{fig:treewidth}
\end{figure}
\end{LONG}

If $T_i$ is a sub-graph of $T$, we use $V_i$ to denote the vertices associated with the nodes in $T_i$, and $G_{V_i}$ to denote the sub-graph of $G$ induced by $V_i$.

\begin{property}\label{prop:td-vertex-separability}\emph{~\cite{KleinbergTardos2005}}
Let $k$ be a node of $T$ and suppose that $T - k$ has components $T_1, T_2, \ldots, T_d$. Then, the sub-graphs 
$G_{V_1 \setminus  X_k},G_{V_2 \setminus  X_k}, \ldots, G_{V_d \setminus  X_k}$
have no vertices in common, and there are no edges between them.
\end{property}

A tree decomposition $(\{X_i \mid i \in I\}, T=(I,F))$ is \emph{nonredundant} if there is no edge $\langle i,j \rangle$ in $T$ such that $X_i \subseteq X_j$. There is a simple procedure to make any tree decomposition be nonredundant without affecting the width: if  there is an edge $\langle i,j \rangle$ in $T$ such that $X_i \subseteq X_j$, one can contract the edge by `folding' the node $i$ into $j$; by repeating this process as often as necessary, one ends up having a nonredundant tree decomposition.

\begin{property}\label{prop:td-nonredundant}\emph{~\cite{KleinbergTardos2005}}
Any nonredundant tree decomposition of an $n$-vertex graph has at most $n$ nodes.
\end{property}

\begin{INUTILE}
\begin{property}\emph{~\cite{Gavril197447}}
Let $W \subseteq V$ be a clique in $G=(V,E)$, and $(\{X_i \mid i \in I\}, T=(I,F))$ be a tree decomposition of $G$. Then there is an $i \in I$ with $W \subseteq X_i$.
\end{property}
\end{INUTILE}

It is \texttt{NP}-hard to determine the treewidth of a given graph. Furthermore, there is no known algorithm to compute a constant-factor approximation of an optimal tree decomposition in polynomial time.  The best polynomial time approximation algorithms for tree decompositions are a $\bigo(\omega \log \sqrt{\omega})$-factor approximation algorithm running in time within $n^{\bigo(1)}$ described by Feige et al.~\cite{Feige2005}; and a $\bigo(\omega \log \omega)$-factor approximation algorithm running in time within $(n^4 \omega \log \omega)$ described by Amir~\cite{Amir10}. 
\begin{CONST-TW}
If the treewidth is known to be constant, then a $(3 \omega + 4)$-approximation can be computed in time within $\bigo(2^{\bigo(\omega)} + n \log n)$,  and a $(5 \omega + 4)$-approximation can be computed in time within $\bigo(2^{\bigo(\omega)} + n )$, using two algorithms described by Bodlaender et al.~\cite{Bodlaender2013} respectively.
\end{CONST-TW}

\subsection{An algorithm sensitive to the intersection graph's treewidth}\label{subsec:treewidth-adapt}

Here we describe an algorithm which benefits from instances that have an intersection graph with small treewidth. We will say that a set of boxes has treewidth $\omega$ if its intersection graph has treewidth $\omega$.

Intuitively, consider a set of $n$ boxes with a tree $T$ as intersection graph. Its \textsc{Klee's Measure} can be computed in a divide-and-conquer fashion: reduce the problem to two sub-problems by dividing the intersection graph, via a vertex removal, into two sub-trees $T_1$, $T_2$ of roughly equal  sizes; solve each problem independently; and then combine their solutions.
 \begin{LONG}
 Since the intersection graph is a tree, one can always find a vertex $v$ that divides the tree into two forests of size at most $\lfloor n/2 \rfloor$;  by adding $v$ back to both forests we obtain $T_1$, $T_2$ of size at most $\lceil (n+1)/2 \rceil$.  The \textsc{Klee's Measure} of the original instance is the sum of the \textsc{Klee's Measure} of 
each sub-instance minus the measure of their intersection, which is the volume of the box (vertex) used to split. 
This procedure yields an algorithm running time within $\bigo(n \log n)$. 
\end{LONG}

This procedure can be extended to the computation of the \textsc{Klee's Measure} of an instance of tree width $\omega$, given a tree decomposition $T$ of its intersection graph.  The following lemma shows how the solutions of two sub-problems can be combined into the general solution. If $t$ is a node of $T$, we denote by $\mathcal{B}_t$ the subset of the boxes of $\mathcal{B}$ corresponding to the vertices within $X_t$.

\begin{lemma}\label{theo:treedecomp-divide-km}
Let $T$ be a tree decomposition of the intersection graph of a set of boxes $\mathcal{B}$, and $t$ be a node of $T$ such that, when removed, $T$ is split into two non-empty sub-forests $F_1$ and $F_2$. Let $T_L = F_1 \cup \{t\}$,  $T_R = F_2 \cup \{t\}$, $\mathcal{B}_L = \bigcup_{l \in T_L}{\mathcal{B}_l}$, and $\mathcal{B}_R = \bigcup_{r \in T_R}{\mathcal{B}_r}$. Then, $T_L$ and $T_R$ are tree decompositions of $\mathcal{B}_L$ and $\mathcal{B}_R$ respectively; and the \textsc{Klee's Measure} of $\mathcal{B}$ equals the \textsc{Klee's Measure} of $\mathcal{B}_L$ plus the \textsc{Klee's Measure} of $\mathcal{B}_R$ minus the \textsc{Klee's Measure} of $\mathcal{B}_t$.
\end{lemma}

\begin{proof}
By Property~\ref{prop:td-vertex-separability} we know that $F_1$ and $F_2$ share no vertices, and there is no edge between them. Hence, no box corresponding to a vertex in a node from $F_1$ can intersect a box corresponding to a vertex in a node from $F_2$. Therefore, the intersection between $\mathcal{B}_L$ and $\mathcal{B}_R $ is  $\mathcal{B}_t$.  This, and the fact that the volume of a box (i.e., the \textsc{Klee's Measure} of a box) is a Lebesgue measure, proves that the  \textsc{Klee's Measure} of $\mathcal{B}_L \bigcup \mathcal{B}_R$ equals the \textsc{Klee's Measure} of $\mathcal{B}_L$ plus the \textsc{Klee's Measure} of $\mathcal{B}_R$ minus the \textsc{Klee's Measure} of $\mathcal{B}_t$. The result follows.
\qed
\end{proof}

Using this lemma, we can apply the procedure described above for trees to general tree decompositions, as in Algorithm~\ref{alg:twadptmeasure}. Lemma~\ref{lemma:treedecomp-divide-km-time} provides an upper bound for the running time of this new solution.

\begin{algorithm}                      
\caption{\texttt{tw\_measure}}
\label{alg:twadptmeasure}
\begin{algorithmic}[1]
\Require A domain box $\Gamma$, and set of $n$ boxes $\mathcal{B}$ in $\mathbb{R}^d$, and an $\rho$-node tree decomposition $T$ of the intersection graph of $\mathcal{B}$
\Ensure The \textsc{Klee's Measure} of $\mathcal{B}$ within $\Gamma$
\If {$\rho = 1$}
	\State let $t \gets$  the only node in $T$	
	\If {the measure of $t$ has not being computed before}
		\State $measures[t] \gets$  \texttt{SDC}($\Gamma, \mathcal{B}_t$)
	\EndIf
	\State  \Return $measures[t]$
\Else
	\State Find a node $t \in T$ that when removed splits $T$ into two sub-forests 
	\Statex $\;\;\;\;\;\{F_1, F_2\}$ of sizes at most $\lceil \rho/2 \rceil$
	\State let $T_1 \gets F_1 \cup \{t\}$, $T_2 \gets F_2 \cup \{t\}$
	\State Let $\mathcal{B}_L \gets \bigcup_{t \in T_L}{\mathcal{B}_t}$, $\mathcal{B}_R \gets \bigcup_{t \in T_R}{\mathcal{B}_t}$
	\State  \Return $\texttt{tw\_measure}(\Gamma, \mathcal{B}_L,  T_L) 
	+ \texttt{tw\_measure}(\Gamma,	\mathcal{B}_R, T_R) - measures[t]$
\EndIf
\end{algorithmic}
\end{algorithm}


\begin{lemma}\label{lemma:treedecomp-divide-km-time}
Let $\mathcal{B}$ be a set of $n$ $d$-dimensional boxes. Let  $T$  be a tree decomposition of the intersection graph of $\mathcal{B}$ with $\rho$ nodes of sizes $n_1, \ldots, n_\rho$, respectively. The \textsc{Klee's Measure} of $\mathcal{B}$ within a given $d$-dimensional domain box $\Gamma$ can be computed in time within $\bigo(\rho \log \rho +  \sum_{i=1}^{\rho}{n_i^{d/2}})$.
\end{lemma}

\begin{proof}
We show that Algorithm ~\ref{alg:twadptmeasure} runs in time within the given bound. The recursion tree corresponding to the algorithm has $\rho$ leaves, and since at each step the size of the problem is approximately reduced by one half, the height is within $\bigo(\log \rho)$. The total running time at each internal level of the tree is within $\bigo(\rho)$, and hence the $\rho \log \rho$ term. Moreover, the \textsc{Klee's Measure} of the boxes within each leaf is computed only once, making the total time of computing the measure of the leaves to be within $\bigo(\sum_{i=1}^{\rho}{n_i^{d/2}})$. The result follows. 
\qed
\end{proof}

When computing the \textsc{Klee's Measure} of a set $\mathcal{B}$ of boxes, for whose intersection graph is given an optimal tree decomposition, the following result follows.

\begin{corollary}\label{col:treedecomp-divide-km-time}
Let $\mathcal{B}$ be a set of $n$ $d$-dimensional boxes, and $T$  be an optimal non-redundant tree decomposition of the intersection graph of $\mathcal{B}$. The \textsc{Klee's Measure} of $\mathcal{B}$ within a given $d$-dimensional domain box $\Gamma$ can be computed in time within $\bigo(n \log n +  n \omega ^{d/2})$, where $\omega$ is the treewidth of the intersection graph of $\mathcal{B}$.
\end{corollary}

\begin{proof}
Let $\rho$ denote the number of nodes in $T$. By Property~\ref{prop:td-nonredundant}, we know that, since $T$ is non-redundant, $\rho \le n$. Besides, each node in the tree decomposition has at most $\omega + 1$ vertices. The result follows by replacing each $n_i$  by $\omega$, for $i=[1..\rho]$, in Lemma~\ref{lemma:treedecomp-divide-km-time}.
\qed
\end{proof}

Note that the bound $\bigo(n \log n +  n \omega ^{d/2})$ in Lemma~\ref{lemma:treedecomp-divide-km-time} is better than or equal to $\bigo(n^{d/2})$ as long as $\omega \le n^{1-2/d}$. For this bound to be achieved, an optimal tree decomposition of the intersection graph is required. This decomposition, in general, needs to be computed, and it is not known whether this can be performed in polynomial time~\cite{KleinbergTardos2005}. 

Optimal tree decomposition of certain classes of graphs can be found efficiently. The \textsc{Klee's Measure} of sets with intersection graph in such classes can be computed in time depending on the treewidth. We describe two examples of such results in Corollaries~\ref{col:kmp_disconnected_graph}~and~\ref{col:kmp_profile_treewidth}.

\begin{corollary}\label{col:kmp_disconnected_graph}
Let $\mathcal{B}=\{b_1, b_2, \ldots, b_n\}$ be a set of $n$ boxes in $\mathbb{R}^d$, $\Gamma$ be a d-dimensional box, and $G$ be the intersection graph of $\mathcal{B}$. The measure of the union of $\mathcal{B}$ within $\Gamma$ can be computed in time within $\bigo(n \log ^{d-1}{n} + \sum_{i=1}^{\rho}{n_i^{d/2}})$, where $\rho$ is the number of connected components of $G$, and $n_1, \ldots, n_\rho$ are the sizes of the $\rho$ connected components.
\end{corollary}

\begin{proof}
An algorithm described by Edelsbrunner et al. ~\cite{Edelsbrunner1984} computes the connected components in time within $\bigo(n \log ^{d-1}{n})$. From the connected components one can easily obtain a tree decomposition of the graph as follows: create a node for each connected component, and add edges between the nodes until obtaining any arbitrary tree. By Lemma~\ref{lemma:treedecomp-divide-km-time} the bound follows.
\qed
\end{proof}

When the profile of the input instance is $k$, the same bound from Lemma~\ref{lemma:kbound} can be achieved by using Algorithm~\ref{alg:twadptmeasure}, as seen in Corollary~\ref{col:kmp_profile_treewidth}.

\begin{corollary}\label{col:kmp_profile_treewidth}
Let $B$ be a set of $n$ boxes in $\mathbb{R}^d$,   $\Gamma$ be a d-dimensional box, and $k$ be the profile of $\mathcal{B}$ within $\Gamma$. The \textsc{Klee's Measure} of $\mathcal{B}$ within $\Gamma$ can be computed in time within $\bigo(n \log n + (n/k)k^\frac{d-2}{2})$.
\end{corollary}

\begin{proof}
We can transform $\mathcal{B}$ into a set $\mathcal{B}'$ with the same \textsc{Klee's Measure}, but with treewidth within $\bigo(k)$, as follows:  
Split the domain into $\bigo(n/k)$ slabs using Algorithm~\ref{alg:splitdomain}. 
Then, for each slab, add to $\mathcal{B}'$ the boxes in $\mathcal{B}$ that intersect the slab, restricted to it (i.e., if a box intersects several slabs, the box will be split into multiple boxes that intersect only one slab, and such that the union of them is the original one)

A tree decomposition of the intersection graph of $\mathcal{B'}$ with $\bigo(n/k)$ nodes of size within $\bigo(k)$ can be obtained as follows: create a node for each slab, and add edges between the nodes until an arbitrary tree is obtained. Since no box can intersect a box out of its slab, the tree decomposition is valid. The bound follows from applying Lemma~\ref{lemma:treedecomp-divide-km-time}.
\qed
\end{proof}

An approximation of the optimal tree decomposition of the input set could be used, obtaining the weaker bounds described in Corollary~\ref{col:aprox-treedecomp-divide-km-time}.


\begin{corollary}\label{col:aprox-treedecomp-divide-km-time}
Let $\mathcal{B}$ be a set of $n$ $d$-dimensional boxes, and $\Gamma$  a $d$-dimensional domain box.
The \textsc{Klee's Measure} of $\mathcal{B}$ within $\Gamma$ can be computed
\begin{INUTILE}
 in time within $\bigo(n \log n)$, when the treewidth $w$ of $\mathcal{B}$ is constant, and 
 \end{INUTILE}
 in time within $\bigo(n^4\omega \log \omega + N (\omega \log \omega)^{d/2})$.
\end{corollary}

\begin{proof}
The result follows by replacing the $\omega$ term in the bound of Lemma~\ref{lemma:treedecomp-divide-km-time} by the
$\bigo(\omega \log \omega)$-factor approximation obtained by Amir's algorithm~\cite{Amir10} (see the end of Section~\ref{subsec:treewidth-preliminaries} for details).
\qed
\end{proof}

Note that the $\bigo(n^4\omega \log \omega + n (\omega \log \omega)^{d/2})$ bound for the general case is lower than $\bigo(n^{d/2})$ as long as $d > 8$ (because of the first term) and $\omega \log \omega \le n^{(d-2)/d}$. 

In the following section, we compare the techniques we have described so far and describe how to combine them.

\section{Combining the Techniques}\label{sec:combining}
A low profile implies that the intersection graph has low treewidth (Corollary~\ref{col:kmp_profile_treewidth}), but a low treewidth does not imply a low profile: an instance of $n$ boxes in the class illustrated in Figure~\ref{fig:degeneracy_tree} has a profile within $\bigo(n)$, and its treewidth is one. 
On the other hand, the running time of the algorithm taking advantage of the profile is never worth than $\bigo(n^{d/2})$, which is not true for the running time of the one sensitive to the treewith, even if an optimal tree decomposition is provided to it.

The treewidth and profile measures are independent from the size of the \textsc{Maxima}. For example, an instance of $n$ boxes in the class illustrated in Figure~\ref{fig:degeneracy_ig} has a \textsc{Maxima} of size $1$, and its treewidth is $n-1$. With respect to the treewidth, this is a `hard' instance, but with respect to the \textsc{Maxima} size is easy. On the contrary, an instance of $n$ boxes in the class illustrated in Figure~\ref{fig:degeneracy_tree} has a \textsc{Maxima} size of $n$, whiles its treewidth is one. 
\begin{INUTILE}
The instance of Figure~\ref{fig:degeneracy_orthants} $\ldots$
\end{INUTILE}
\begin{figure}
\centering
\begin{subfigure}{.45\textwidth}
  \centering
  \includegraphics[height=3.8cm]{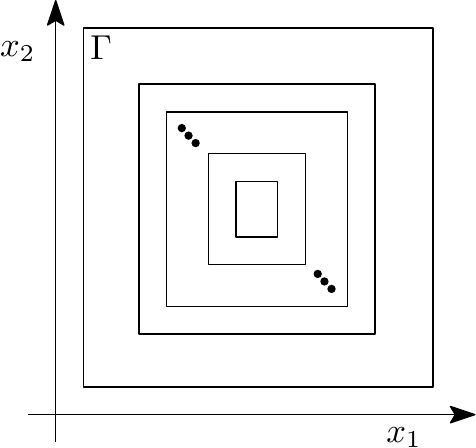}
  \caption{ }
  \label{fig:degeneracy_ig}
\end{subfigure}
\begin{subfigure}{.45\textwidth}
  \centering
  \includegraphics[height=3.8cm]{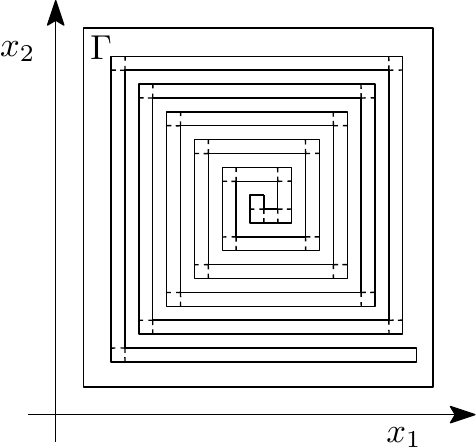}
  \caption{ }
  \label{fig:degeneracy_tree}
\end{subfigure}
\caption{Two classes of instances of the \textsc{Klee's Measure} problem that are ``easy'' or ``hard'' depending on the measure considered: instance (a) is easy if the size of the \textsc{Maxima} is considered, but difficult  if the profile or treewidth are considered; instance (b) is easy for treewidth, but hard for both \textsc{Maxima} size and profile.} 
\label{fig:degeneracy_comp} 
\end{figure}

Since these two measures are independent, we can combine them to obtain an algorithm sensitive to both of them, at the same time, by computing the \textsc{Maxima} of the set; and finding the \textsc{Klee's Measure} of the \textsc{Maxima} of the remaining graph as described in Corollary~\ref{col:aprox-treedecomp-divide-km-time}. This way of proceeding yields an algorithm with running time improving over the results from Lemmas~\ref{lemma:kmp_boxes} and~\ref{lemma:treedecomp-divide-km-time}.
\begin{theorem}\label{theo:combined}
Let $\mathcal{B}$ be a set of $n$ boxes in $\mathbb{R}^d$, $\Gamma$ a $d$-dimensional box, and $G$ be the intersection graph of $\mathcal{B}$. The \textsc{Klee's Measure} of $\mathcal{B}$ within $\Gamma$ can be computed in time within 
$\bigo \left(n \log^{2d-2} h + h^4 \omega \log \omega + h(\omega \log \omega)^{d/2} \right)$, 
where $h$ is the size of the \textsc{Maxima} $M(\mathcal{B})$ of $\mathcal{B}$, and $w$ its treewidth.
\end{theorem}

No lower bound is known for this problem with respect to these measures. In fact, we believe the results obtained here can be further improved, by considering finer versions of the measures in order to improve the analysis. We describe 
preliminary results in this direction in the next section.

   \section{Discussion}\label{sec:discussion}

Each of the three boosting techniques that we analyzed can be improved, and we describe preliminary results for each technique in those directions below, as well as other lines of research.

The \textsc{Maxima} based technique (described in Section~\ref{sec:maximaFiltering}) yields an algorithm running in time within $\bigo(n(\log h)^{2d-2} + h^{d/2})$, where $h$ is the size of the \textsc{Maxima} of the input set. This bound can be improved. For example, if instead of filtering items not in the \textsc{Maxima} only once, this is done as part of the \emph{simplification} step of the algorithm \emph{Simplify, Divide and Conquer} (\texttt{SDC})~\cite{Chan2013}, the expression for its running time becomes $T(n)=2T(\frac{h}{2^{2/d}}) + \bigo(n \log^{2d-2}{h})$. This running time is still within $\bigo(n^{d/2})$ in the worst case, and also within $\bigo(n \log^{2d-2}{h}+h^{d/2})$. It is never worse (asymptotically) and it is better in many cases, but how to formally analyze this improvement is still an open question.

The profile based technique (described in Section~\ref{sec:intrinsic}) yields a solution running in time within $\bigo \left((n-k) \log \frac{n-k}{k} + n k^{\frac{d-2}{2}}  \right)$. The algorithm \texttt{SDC}~\cite{Chan2013} is already adaptive to the profile of the input set. It has a limitation though: it necessarily cycles over the dimensions in order to ensure running in time within $\bigo(n^{d/2})$. If there are few dimensions where the profile of the set is small, this technique performs considerably better than \texttt{SDC}. The technique could be further improved if, instead of considering an upper bound for the profile in each sub-problem, we use the exact value of the profile of the subproblem. However, it is not clear how to analyze this improvement.

The treewidth based technique (described in Section~\ref{sec:intersectionGraphTreewidth}) yields an algorithm running in time within $\bigo(n \log n +  n \omega ^{d/2})$, if an optimal tree decomposition of the intersection graph is given; and in time within $\bigo(n^4\omega \log \omega + n (\omega \log \omega)^{d/2})$ if not. Note that these running time bounds are not always better than or equal to the $\bigo(n^{d/2})$ achivied by algorithm \texttt{SDC}. The dependence on a tree decomposition is the main weakness of this approach. It is not known whether the \textsc{Klee's Measure} can be computed by using a treewidth-sensitive algorithm that does not depend explicitly on a tree decomposition of the intersection graph. 

Finally, note that the techniques that we described focus on the structure of the instance, as opposed to the order in which the instance is given. As such, the algorithms that we described cannot beat the $\bigo(n \log n)$ bound, even though there are instances that can still be solved in time within $\bigo(n)$. If one considers the order in which the input is given, for special pre-sorted inputs one can achieve the $\bigo(n)$ bound.

\begin{LONG}
\section*{Acknowledgement}
All authors were partially supported by Millennium Nucleus Information and Coordination in Networks ICM/FIC RC130003. We thank Timothy Chan for his helpful comments, and one anonymous referee from ESA 2015 for pointing out the relation between our techniques of analysis and the treewidth of the intersection graph.
\end{LONG}
